\documentclass[a4paper,12pt]{article}
\pdfoutput=1 
\usepackage{geometry}
\usepackage[margin=10mm]{caption}
\usepackage[utf8]{inputenc}
\usepackage[english]{babel}
\usepackage{amssymb, amsmath, amsthm, amsfonts, eucal, mathalpha, mathrsfs, yfonts, lmodern, slashed, microtype}
\usepackage[sorting=none, backend=biber, style=numeric-comp, natbib=true]{biblatex}

\usepackage{array, booktabs}
\usepackage{graphicx, subcaption, tikz, tikz-cd}
\usetikzlibrary{angles, quotes, calc, shapes.geometric}
\usepackage{authblk, csquotes, hyperref}

\theoremstyle{definition}
\newtheorem{definition}{Definition}[section]

\newtheorem{theorem}{Theorem}[section]

\newtheorem{conjecture}{Conjecture}[section]
\newcommand{\beq}{\begin{equation}}
\newcommand{\eeq}{\end{equation}}
\newcommand{\bqa} {\begin{eqnarray}}
\newcommand{\eqa} {\end{eqnarray}}

\def\mfsl{\mathfrak{sl}}
\def\mfg{\mathfrak{g}}
\def \mmV{\mathtt{V}}
\def \mC {\mathbb{C}}
\def \la{\lambda}
\def \mL{\mathtt{L}}
\def\mM {\mathtt{M}}
\def \mZ {\mathbb{Z}}
\def\ga{\gamma}
\def\tr{\mathtt{tr}}
\def\V{\mathscr{V}}
\def\mR{\mathbb{R}}
\def\M{\mathscr{M}}
\def\mN{\mathbb{N}}
\def\C {\mathscr{C}}
\def\al{\alpha}
\def\bet{\beta}

\begin{document}
\begin{flushright}
    ITEP-TH-24/23
\end{flushright}
\vspace{1cm}
\begin{center}
{\Large{\sf  Generalized theta series and monodromy of Casimir connection. The case of rank 1

 }
}
\vspace{10mm} {\sf E. Dotsenko$^{\star}$  
 \\ \vspace{7mm}
 \vspace{2mm} $\star$ - {\sf NRC ''Kurchatov Institute'', 123182,  Moscow, Russia}\\

 \vspace{4mm}
 {\footnotesize \sf email:  edotsenko95@gmail.com\\}
 \vspace{2mm}
\vspace{2mm}}

\end{center}
\begin{abstract}
The monodromy of the $\mfsl(2)$ Casimir connection is considered. It is shown that the trace of the monodromy operator over the appropriate space of flat sections gives rise to the Jacobi theta constant and to the partial Appell-Lerch sums. 

\end{abstract}
\vspace{2mm}
\begin{center}
    \emph{Dedicated to my parents.}
\end{center}
\vspace{2mm}
\begin{center}
\end{center}
\vspace{-4mm}
\section{Introduction}
In the paper the monodromy of the (deformed) $\mfsl(2)$ Casimir connection is studied. Traces of the monodromy of the Casmir connection and its deformation over the Verma modules give rise to the classical Jacobi theta-constant and partial theta functions correspondingly. We recall the definition of the partial Appell-Lerch sums. It is shown that these sums are traces of the monodromy operator of the Casimir connection over the tensor products of the Verma modules. It is remarkable that the similar formulas do appear in the modern 3d topological field theories  \cite{1}, \cite{2}, \cite{3} and 2d conformal field theories \cite{4}, \cite{5}. Les us begin with the short introduction into the subject of the Casimir connection.\par
 To each finite dimensional semisimple complex Lie algebra one can associate two flat connections - the Knizhnik-Zamolodchikov (KZ) flat connection $\nabla_{\mathtt{KZ}}$ and the Casimir connection $\nabla$. The KZ connection appears naturally in the WZW model \cite{6}. The Casimir connection is defined by its compatibility with the KZ connection and shifts the twist parameters of the $\nabla_{\mathtt{KZ}}$ \cite{7}, \cite{8}, \cite{9}. \par   
   Historically the theory of the Casimir connection was built by De-Concini, Felder-Markov-Tarasov-Varchenko \cite{8} and, independently by Millson and Toledano-Laredo \cite{10}, \cite{11}. The generality of \cite{8} allows one to consider the Casimir connection associated to the Kac-Moody algebra (without Serre relations) $\mfg(\mathtt{C})$, where $\mathtt{C}$ is the Cartan matrix of general type. Let us also recall the recent developments \cite{12}.\par  
   The material is organized in the following way: in the \textbf{second section} the relevant facts about $\mfsl(2)$  and representations $\mmV$ are recalled, which will be used to get the Jacobi theta-constant and partial Appell-Lerch sums in the parts 3 and 4 correspondingly. In the \textbf{third section} the connection between $\mmV$, Jacobi theta-constant and partial theta functions is established. In the \textbf{fourth section} the definition of the partial Appell-Lerch sums and their connection to the tensor products of the representations from the second part is presented. 

\section{The $\mfsl(2)$ recollections}
In this part the definition of the $\mfsl(2)$ algebra is recalled, the truncated Casimir operator is defined, also the list of the $\mfsl(2)$ representations is given. As an example one considered representation $\tt{P}$ in which the Jordan type of the truncated Casimir operator is non-trivial.

\subsection{Definition}
The Lie algebra is defined as  $\mfsl(2) = \mC e\oplus \mC h \oplus \mC f , $
 
\begin{equation}
\begin{gathered}
 [h,e] = 2e, \,\, [h,f] = -2f, \,\, [e,f] = h. 
\end{gathered}
\end{equation}
The truncated Csimir operator is the following element of the universal enveloping algebra $\tt{U}(\mfsl(2))$

\begin{equation}
\kappa = ef + fe.
\end{equation}
On $\tt{U}(\mfsl(2))$ there is a comultiplication
 
\begin{equation}
    \begin{gathered}
        \delta: \tt{U}(\mfsl(2))\to \tt{U}(\mfsl(2))\otimes \tt{U}(\mfsl(2)), \\
        \delta(a) = a\otimes 1+ 1\otimes a,\,\, a\in \mfsl(2), \\  \label{comult}
        \delta(xy) = \delta(x)\delta(y),\,\, x,y\in  \tt{U}(\mfsl(2)).
    \end{gathered}
\end{equation}
In particular
\begin{equation}
    \delta(\kappa) = \left(ef + fe\right)\otimes 1  + 2\left(e\otimes f+ f\otimes e\right) + 1\otimes \left(ef + fe\right). \label{comultka}
\end{equation}

\subsection{The list of representations}
Consider the following list of $\mfsl(2)$ representations $\tt{V}$:
\begin{equation}
    \mathtt{M}_{\la},\,\,\mathtt{L}_{\la},\,\, \,\, \mathtt{P}, \label{osl2listprincipal}
\end{equation}
where $\la\in\,0,-2$. For $\la\in\mC$  the Verma module $\mM_{\la}$ has a basis $\{ f^{k}v_{\la}\}_{k = 0}^{\infty}$. Highest weight vector is defined as 
\begin{equation}
e v_{\la} = 0, \\
h v_{\la} = \la v_{\la}.
\end{equation}
 $\mL_\la$ -  irreducible factor of the $\mathtt{M}_{\la}$.

 Representation $\mathtt{P}$ is the tensor product 
\begin{equation}
    \mathtt{P} = \mM_{-1}\otimes \mL_1.
\end{equation}
Representations $\eqref{osl2listprincipal}$ are distinguished by the fact that the Casimir element acts on each weight subspace with zero eigenvalues.

\begin{theorem}
\label{jordancasimir}
 $\kappa$ acts on the weight subspace $\mathtt{P}(-2k)$, $k>0$ and has the following Jordan type 
\begin{equation}
    \kappa\vert_{\mathtt{P}(-2k)} = \begin{pmatrix}
        -2k^2&2 \\
        0& -2k^2
    \end{pmatrix}. 
\end{equation}

\end{theorem}

\begin{proof}
$\kappa$ acts on the weight subspace because 
\begin{equation}
    [h,\kappa] = 0.
\end{equation}
 $\mathtt{L}_1 = \mC u_{1}\oplus \mC u_{-1},$ where $u_{\pm1}$ are the vectors of highest/lowest weight correspondingly. Because   $ \mathtt{P} = \mM_{-1}\otimes \mL_1$ is the tensor product, hence the action of the truncated Casimir element acts via $\eqref{comultka}$ on weight subspace

 $$\mathtt{P}(-2k) = \mC f^{k}v_{-1}\otimes u_{1}\oplus \mC f^{k-1}v_{-1}\otimes u_{-1}$$  has the form

\begin{equation}
    \left(\kappa\right)\vert_{\mathtt{P}(-2k)} = \begin{pmatrix}
        -2k(k+1) & 2 \\
        - 2k^2 &  -2k(k-1)
    \end{pmatrix}.
\end{equation}
This operator could be brought to the Jordan normal form  
\begin{equation}
\begin{gathered}
     \left(\kappa\right)\vert_{\mathtt{P}(-2k)} =  \begin{pmatrix}
        -2k^2 & 2 \\
        0 &  -2k^2
    \end{pmatrix}, \,\, \text{in the basis } \\
    \{f^{k}v_{-1}\otimes u_{1} + kf^{k-1}v_{-1}\otimes u_{-1},\, f^{k-1}v_{-1}\otimes u_{-1}\}.
\end{gathered}
\end{equation}

\end{proof}
\section{Jacobi  theta- constant and partial theta series}
In this part the Casimir connection $\nabla$ is defined and the computation of the trace of the monodromy over the spaces of flat sections, that corresponds to the representations $\eqref{osl2listprincipal}$, as particular trace the Jacobi theta-constant appears. Also the deformed Casimir connection $\nabla(\phi)$ is defined and 
 traces over the same spaces are computed for it. The main result of this part is the Theorem 1 and formulas $\eqref{theta}$. 
\par
Casimir connection is defined on a trivial bundle
\begin{equation}
    \V = \mC^{\times}\times \tt{V}, \label{triv}
\end{equation}
where $\tt{V}$ is a representation from the list $\eqref{osl2listprincipal}.$ 
The connection and the equation on the flat sections have the following form
\begin{equation}
\begin{gathered}
\nabla = dz\left(\frac{d}{dz} + \hbar \frac{\kappa}{z}\right), \\
     \nabla \psi(z) = 0. \label{cassl2}
\end{gathered}
\end{equation}
where $z$ is a coordinate on $\mC^{\times}$, $\hbar\in \mC$ , $\psi(z)$  is a multivalued function with 0 as its singularity.  $\Gamma(\tt{V},\nabla)$ denotes the space of the flat sections of the flat connection. There is an action of the monodromy group on this space.
\begin{equation}
\pi_{1}(\mC^{\times}) = \mZ \curvearrowright \Gamma(\tt{V},\nabla),
\end{equation}
defined as the analytic continuation  $\mathscr{M}_{\ga}$ along the path $\ga\in \pi_{1}(\mC^{\times})$
\begin{equation}
    \mathscr{M}_{\ga}\cdot\psi(z) = \psi(\ga(1)z).
\end{equation}
The generator in the fundamental group is chosen in the following way
\begin{equation}
 \ga_{0}(t) = \exp(2\pi i t)z_0,\,\, z_{0}\in\mC^{\times},\,\,  t\in [0,1].
\end{equation}
In what follows we will need the following notation:
 
\begin{enumerate}
    \item  $q = \exp(4\pi i \hbar )$
    \item  The trace of the monodromy operator over the space of flat sections
\begin{equation}
  \chi(q,l, \tt{V}): = \tr_{\Gamma(\tt{V},\nabla)} \mathscr{M}^l_{\ga_{0}},\,\, l\in \mZ.
\end{equation}
\item $ \Gamma(\tt{V}(-2k),\nabla)$ - the space of flat sections $\psi(z) $ with values in $\mmV(-2k)$.  
\item $\mathscr{M}_{\ga_{0}}\vert_{\Gamma(\tt{V}(-2k),\nabla)}$ - monodromy operator, restricted to  $\Gamma(\tt{V}(-2k),\nabla)$.
\end{enumerate}

\begin{theorem}
\label{sl2jacteta}
Let $C_{1},C_{2}\in \mC.$ 
Flat sections, their monodromy and the corresponding theta series are collected in the Table 1
 
\begin{table}[ht]
    \centering
    \begin{tabular}{c|c|c|l}
       Representation  &  $\Gamma(\tt{V}(-2k),\nabla)$ &  $\mathscr{M}_{\ga_{0}}\vert_{\Gamma(\tt{V}(-2k),\nabla)}$ &\,\,\,\,\,\,\,\,\,\, $ \chi(q,l, \tt{V})$ \\ \hline
       & & & \\
        $\mL_0$ & $C_1$& 1 & \,\,\,\,\,\,\,\,$\chi(q,l,\tt{L}_{0}) = 1$ \\
        & & & \\
        $\mM_{\lambda}$, $\la = 0,-2$ & $\begin{cases}C_{1}z^{ 2k^2\hbar} ,\,\, k\in -2\mZ_{+}\\
        C_{1}z^{ 2k^2\hbar} ,\,\, k\in -2\mZ_{\geq 0}-2
        \end{cases}
        $ & $\begin{cases}q^{ k^2},\,\, k\in -2\mZ_{+}\\
        q^{ k^2} ,\,\, k\in -2\mZ_{\geq 0}-2
        \end{cases}$ & 
        \,\,\,\,\,\,$\begin{cases}
          \displaystyle \chi(q,l,\mM_{0}) = \sum_{k\geq 0} q^{l  k^2} \\
           \displaystyle \chi(q,l,\mM_{-2}) = \sum_{k\geq 1} q^{l  k^2}
        \end{cases}$ \\
        & & & \\
        $\mathtt{P} $ & $\begin{pmatrix}z^{2k^2\hbar}\left(C_1 -2 \hbar C_2 \log(z) \right) \\ C_2z^{2k^2\hbar} \end{pmatrix}$ &  $\begin{pmatrix}q^{k^2} & 0 \\ 0 & q^{k^2}\end{pmatrix}\begin{pmatrix}1 & -4 \pi i \hbar\\ 0 & 1\end{pmatrix}$ & \,\,\,\,\,\,\,\,\,$\chi(q,l,\tt{P}) = \displaystyle \sum_{k\in \mZ} q^{l k^2} $
    \end{tabular}
    \caption{}
    \label{tab:conn_data_sl2} 
\end{table}
\end{theorem}
\begin{proof}
Let us consider the computation of the flat sections with values in $\mathtt{P}$.

$\psi(z)\in \Gamma(\tt{P}(-2k),\nabla)$ has the form
\begin{equation}
    \psi(z)\vert_{\mathtt{P}(-2k)} = \psi_1(z)f^{k}v_{-1}\otimes u_{1} + \psi_2(z)f^{k-1}v_{-1}\otimes u_{-1}.
\end{equation}
Using the Theorem 1, one has

\begin{equation}
\begin{gathered}
    \left(\frac{d}{dz} + \hbar \frac{\begin{pmatrix}
        -2k^2 & 2 \\
        0 &  -2k^2
    \end{pmatrix}}{z}\right) \begin{pmatrix}
        \psi^{'}_1(z)\\ \psi^{'}_2(z)
    \end{pmatrix} = 0, \\
    \begin{pmatrix}
        \psi^{'}_1(z)\\ \psi^{'}_2(z) \end{pmatrix} = \begin{pmatrix}
        z^{2k^2\hbar}\left(C_1 - 2\hbar C_2 \log(z) \right) \\ C_2z^{2k^2\hbar} \end{pmatrix}.
\end{gathered}
\end{equation}

\par 
Jacobi theta-constant in this approach is 
 $\chi(q,l, \tt{P})$. Note that if $l>0$, then the trace converges in the unit disc, otherwise the domain of convergence is the exterior $|q|>1$. Other lines of the Table 1 are considered similarly.  One has to note the relationship with the Riemann zeta function 
\begin{equation}
\begin{gathered}
        \int_{i\mR_{+}}\frac{d\hbar}{\hbar}\chi(e^{4\pi i\hbar},l, \tt{M}_{-2}))\hbar^{\frac{s}{2}} = \frac{\Gamma(\frac{s}{2})}{\left( 4\pi i l \right)^{\frac{s}{2}}}\sum_{n = 1}^{\infty}\frac{1}{n^{s}} 
        = \frac{\Gamma(\frac{s}{2})}{\left( 4\pi i l \right)^{\frac{s}{2}}}\zeta(s).
\end{gathered}
\end{equation}
The r.h.s. of the formula as a function of the \emph{natural number} $l$ possesses the analytic continuation to the multivalued function  on $\mC^{\times}$, at the same time the integral in the l.h.s. could diverge. \par 

\end{proof}
On the bundle $\eqref{triv}$ the deformed Casimir connection $\nabla(\phi), \,\phi\in \mC$ is given by the following formula
\begin{equation}
    \nabla(\phi) = dz\left(\frac{d}{dz} + \hbar \frac{\kappa}{z} +\phi \frac{ h}{z}\right).
\end{equation}
The monodromy operator of $\nabla(\phi)$ around 0 is denoted as $\M_{\ga}(x), x = e^{2\pi i \phi},$ its trace over the space of the flat sections $\Gamma(\tt{V},\nabla (\phi))$
\begin{equation}
    \chi(q,x,l, \mathtt{V} ) :=  \tr_{\Gamma(\mathtt{V},\nabla (\phi))}\M_{\ga}(x).
\end{equation}
It is easy to see that for the deformed connection the last column of the Table 1 is the following 

\begin{equation}
    \begin{gathered}
  \chi(q,x,l, \tt{L}_0 )  = 1, \,\,\,\,\,\,\, \chi(q,x,l, \mM_{0} )  =  \sum_{k\geq 0} q^{l  k^2} x^{lk},\\
  \chi(q,x,l, \mM_{-2} )   =   \displaystyle \sum_{k\in \mN} q^{l  k^2} x^{lk}, \,\,\,\,\,\,\, \chi(q,x,l, \tt{P} )  =  1+ 2\displaystyle \sum_{k\in \mN} q^{l k^2}x^{lk} \label{theta}.
    \end{gathered}
\end{equation}
These traces are known in the literature as partial theta functions. Their modular properties were studied in \cite{5}. Such partial theta functions appear in \cite{5} through the 2d conformal field theory. In the present work such  theta series appear independently.

\section{Partial Appell-Lerch sums}
In this section the definition of the partial Appell-Lerch sums is given, and the connection of this sums with the tensor products from the list $\eqref{osl2listprincipal}$ is shown - Theorem \ref{appelllerch}. The section is concluded by the Conjecture 1.

   Let $q,$ $x_{1},\,x_{2}$ be a formal variables, the following theta-series are known as the Appell-Lerch sum \cite{13}
\begin{equation}
    \begin{gathered}
                f(q,x_{1},x_{2}) =  \sum_{n\in \mZ}\frac{q^{n^2}x_{1}^{n}}{1-q^{2n}x_{2}} =  \sum_{v\in \C} q^{v^{t}\tt{B}v}x_{1}^{v_{1}}x_{2}^{v_{2}},  \\
        \tt{B} = \begin{pmatrix}
            1& 1 \\
            1 &0 
        \end{pmatrix},
    \end{gathered}
    \end{equation}
where $v = \begin{pmatrix}
    v_{1}\\
    v_{2}
\end{pmatrix}\in \mZ^{2}$, $\C\subset \mZ^{2}$ is the set of vectors which have non-negative second coordinate $v_{2}\geq 0.$

\begin{definition}
  Let $\al_{i}, \bet_{i}\in\mZ _{+}, i = 0,\ldots, p;\,\, p  > 0.$ Partial Appell-Lerch sum is defined as the following expression 
    \begin{equation}
        \tilde{f}(q,\overline{\al}; \overline{\bet}; p) = \sum_{n\in \mZ_{\geq 0}}q^{n^2}\frac{\prod_{i = 0}^{p}(\al_i + \bet_i q^{(2n+1)})}{(1-q^{(2n+1)})^p} \label{appelllerchinc}
    \end{equation}
    
    \end{definition}
    It is worth to mention that partial theta functions, associated to bilinear forms of the higher rank was studied in \cite{4}, but partial Appell-Lerch sums were not studied there.
\newpage

\begin{theorem}
\label{appelllerch}
Le us denote $\mathtt{R}_{i} = \mM_{0}^{\oplus\al_i}\oplus \mM_{-2}^{\oplus\bet_i}$. Correspondence of $\mathtt{R}_{i}$ to the partial Appell-Lerch sums has the shape outlined in the Table 2
\begin{table}[h!]
\centering
    \begin{tabular}{c|l}
       Representation   & \,\,\,\,\,\,\,\,\,\, \,\,\,\,\,\,\,\,\,\, \,\,\,\,\,\,\,\,\,\, $ \chi(q,l,\tt{V})$ \\ \hline
         & \\
        $\displaystyle\mM_{0}\otimes \mM_{0}$ & $\chi(q,l,\mM_{0}\otimes \mM_{0}) = \displaystyle\sum_{n\in \mZ_{\geq 0}}\frac{q^{ ln^2}}{1-q^{ l(2n+1)}}$ \\
         & \\
        $\displaystyle\mM_{0}\otimes \mathtt{P}$  & $\chi(q,l,\mM_{0}\otimes \mathtt{P}) = \displaystyle \sum_{n\in \mZ_{\geq 0}}q^{ln^2}\frac{1+q^{l(2n+1)}}{1-q^{l(2n+1)}}$\\ 
        & \\
        $\displaystyle\mathtt{P} \otimes \mathtt{P}$ & $ \chi(q,l,\mathtt{P}\otimes \mathtt{P}) = \displaystyle\sum_{n\in \mZ_{\geq 0}}q^{ln^2}\frac{\left(1+q^{l(2n+1)}\right)^2}{1-q^{l(2n+1)}}$ \\
        &  \\
        $\displaystyle\bigotimes_{i = 0}^{p}\mathtt{R}_{i}$ & $ \chi(q,l,\bigotimes_{i = 0}^{p}\mathtt{R}_{i}) = \tilde{f}(q,\overline{\al}; \overline{\bet}; p)$
    \end{tabular}
    \caption{}
    \label{tab:conn_data_sl2_appel}
\end{table}

\end{theorem}
\begin{proof}
 
In order to compute the traces of the monodromy of the Casimir connection, let us utilize the following decomposition 
\begin{equation}
\begin{gathered}
 \mM_{0}\otimes \mM_{0} =  \bigoplus_{i = 0}^{\infty}\mM_{-2i}. \label{sl22vermatensorprod1} 
\end{gathered}
\end{equation}
Using the description of the flat sections and their monodromy with values in $\mM_{-2i}$ one has the following trace

\begin{equation}
    \mathtt{tr}_{\Gamma(\mM_{0}\otimes \mM_{0},\nabla)}\mathscr{M}_{\ga_0} = \sum_{n,m\geq 0}q^{ n^2 + (2n+1)k } = \sum_{n\geq 0}\frac{q^{n^2}}{1-q^{2n+1}},
\end{equation}
that is the partial Appell-Lerch sum.  \par
Let us consider the general tensor product $\displaystyle\bigotimes _{i = 0}^{p}\mathtt{R}_{i}$. There is a sequence of non-negative integers $\mathtt{a}_{k,p}$ such that 
\begin{equation}
    \bigotimes _{i = 0}^{p}\mathtt{R}_{i} = \bigoplus_{k = 0}^{\infty} \mM_{-2k}^{\oplus\mathtt{a}_{k,p}}.
\end{equation}
Let us compute the character for the both sides
\begin{equation}
    \begin{gathered}
        \chi'\left(\bigotimes _{i = 0}^{p}\mathtt{R}_{i},x \right) = \prod_{i  = 0}^{p}\frac{\al_{i} + \bet_{i}x^{-2}}{1-x^{-2}} = \frac{1}{1-x^{-2}}\sum_{k = 0}^{\infty}\mathtt{a}_{k,p}x^{-2k}, \,\,\, \text{hence} \\
        \sum_{k = 0}^{\infty}\mathtt{a}_{k,p}x^{-2k} = \frac{\prod_{i  = 0}^{p}\left(\al_{i} + \bet_{i}x^{-2} \right)}{\left(1-x^{-2}\right)^p}
    \end{gathered}
\end{equation}
Thus,

\begin{equation}
\begin{gathered}
    \mathtt{tr}_{\Gamma(\displaystyle\bigotimes _{i = 0}^{p}\mathtt{R}_{i},\nabla)}\,\mathscr{M}_{\ga_0} = \sum_{n,m\geq 0}q^{n^2 + (2n+1)k } \mathtt{a}_{k,p} = \\ = \sum_{n\geq 0}q^{n^2}\frac{\prod_{i = 0}^{p}\left(\al_{i}+\bet_{i}q^{2n+1}\right)}{\left(1-q^{2n+1}\right)^{p}}. \label{generalsl2trace}
\end{gathered}
\end{equation}
\end{proof}
Let us put forward the following 
\begin{conjecture}
    Let us consider the following tensor product $\displaystyle\bigotimes _{i = 0}^{p}\mathtt{F}_{i},$ where $\mathtt{F}_{i} = \mM_{0}^{\oplus \al_{i}}\oplus\mM_{-2}^{\oplus \bet_{i}}\oplus \mathtt{P}^{\oplus\gamma_{i}}$. Let us denote $\mathtt{F}'_{i} = \mM_{0}^{\oplus \al_{i} + \gamma_{i}}\oplus\mM_{-2}^{\oplus \bet_{i}+ \gamma_{i}}$. Then one has the equality of partial theta series
    \begin{equation}
        \chi(q,l,\bigotimes _{i = 0}^{p}\mathtt{F}_{i} ) = \chi(q,l,\bigotimes _{i = 0}^{p}\mathtt{F}'_{i}).
    \end{equation}
\end{conjecture}
\section{Conclusion}
In the paper the connection was demonstrated among the trace of the monodromy of the $\mfsl(2)$ Casimir connection, Verma modules and partial theta functions; tensor products of the Verma modules and partial Appell-Lerch sums.  A generalization of this observation to the Lie algebras of higher rank is planned to discuss in the follow-up papers. 

\vspace{10mm}
\textbf{Acknowledgments.} 
The author thanks B. Feigin, S. Gukov, A. Levin, and M. Olshanetsky for their
interest in the work and support. The author is also grateful to the referee for the valuable comments. \par 
\textbf{Funding.} This work was supported by ongoing institutional funding. No additional grants to carry
out or direct this particular research were obtained. \par
\textbf{Conflict of interest.} The author of this work declares that he has no conflicts of interest.

\end{document}